\documentclass[11pt,aps]{amsart}
\usepackage{graphicx}
\usepackage{amsmath,amssymb,amscd} 
\usepackage[all]{xy}
\newcommand{\Arg}[1]{\mathrm{Arg}(#1)}
\newcommand{\C}{\mathbb{C}}
\newcommand{\Z}{\mathbb{Z}}
\newcommand{\R}{\mathbb{R}}

\newtheorem{proposition}{Proposition}
\newtheorem{lemma}{Lemma}
\input{Qcircuit}
\begin{document}
\title{A quantum algorithm for Viterbi decoding of classical convolutional codes}
\author{Jon R. Grice}
 \author{David A. Meyer}
 \email{jgrice, dmeyer@math.ucsd.edu}
 \date{\today}
\begin{abstract}
	We present a quantum Viterbi algorithm (QVA) with better than classical performance under certain conditions. In this paper the proposed algorithm is applied to decoding classical convolutional codes, 
for instance; large constraint length $Q$ and short decode frames $N$. Other applications of the classical Viterbi algorithm where $Q$ is large (e.g. speech processing) could experience
significant speedup with the QVA. The QVA exploits the fact that the decoding trellis
	is similar to the butterfly diagram of the fast Fourier transform, with its corresponding fast quantum algorithm. The tensor-product structure of the butterfly diagram corresponds to a
	quantum superposition that we show can be efficiently prepared. The quantum speedup is possible because the performance of the QVA depends on the fanout (number of possible transitions from
	any given state in the hidden Markov model) which is in general much less than $Q$.
The QVA constructs a superposition of states which correspond to all legal paths through the decoding lattice, with phase 
a function of the probability of the path being taken given received data. A specialized amplitude amplification procedure is applied one or more
times to recover a superposition where the most probable path has a high probability of being measured.

\end{abstract}
\maketitle
\section{Introduction}
\label{sect:intro}
	In his 1971 paper on the Viterbi Algorithm (VA), Forney \cite{forney} noted that
	``Many readers will have noticed that the trellis reminds them of the computational flow diagram of
the fast Fourier transform (FFT). In fact, it is identical, except for length, and indeed the FFT is also ordinarily
organized cellwise. While the add-and-compare computations of the VA are unlike those involved in the FFT, some of
the memory-organization tricks developed for the FFT may be expected to be equally useful here.''

The celebrated quantum Fourier transform \cite{deutsch}, \cite{copper} uses tensor-product structure to achieve its
increase in efficiency over its classical counterpart, and so in this paper we propose a quantum Viterbi algorithm (QVA)
taking advantage of the tensor-product structure of the trellis that in certain applications may outperform the (classical)
VA.

Amongst many other applications, the  VA \cite{viterbi} is useful in the decoding of convolutional codes. Thus, in
this paper we will demonstrate the application of a quantum algorithm to decoding classical convolutional codes.
	Quantum algorithms for decoding simplex codes \cite{bargzhou}, more generally Reed-Muller codes \cite{mont} have been proposed and show improvements over the classical algorithms. Protecting quantum information with convolutional encoding has been discussed in \cite{chau,oll}. Applying Grover's algorithm to the decoding of convolutional codes has been demonstrated in \cite{moh},
	which shares some superficial features to the QVA in this paper.
	
For a finite alphabet $Q$,
the VA \cite{forney} is a way to find the most likely
	sequence of states $\pi^* \in Q^N$ that a hidden Markov model (HMM) transitions through given a set of observations $Z$, called emissions. The maximum number of paths branching away from a given state in a single step of the process will be defined to be the fanout $F$.

A single iteration of the QVA will be shown to have gate complexity $O(N|Q|F(\log F)^2)$ and time complexity $O(N \log F)$,
if $|Q|$ quantum systems can be manipulated simultaneously.

Even if $F$ is not considerably smaller than $|Q|$, the gates involving the $|Q|$ complexity factor can be performed in parallel.
However $F$ is often a fixed number within a family of HMMs indexed by size $|Q|$, for example a large graph ($|Q|^2$ very large) where every node is connected to a few ($F$) of its neighbors. $k$-order HMMs can be recast into a first order HMM with large
$|Q|$ and typically low $F$.

The first step of the QVA is to prepare a lattice of paths through the HMM with each path corresponding to a quantum state.
The states can be preloaded with probability amplitudes at the same time in some variants of the QVA. Then the
phases of the states are marked with a function depending on the probability of the path occurring. This can also
be done simultaneously with the lattice building and take advantage of the tensor-product structure of the lattice.
Finally a function-optimization algorithm is called to extract the path probabilities from the
relative phases of the paths and convert them to probability amplitudes. In this paper, the optimization algorithm is
a variant of amplitude amplification. The phase marking step and the amplification step can be repeated some number
of times to make the most probable path more likely to be observed, and we will call those steps together an {\it iteration}
of the algorithm.

The number of iterations in general can be up to $O(\sqrt{L})$ as in Grover's algorithm \cite{grover}, where 
here $L =N F \log F$. This means that in addition to the parallelism in $|Q|$, the QVA can be used with advantage 
when $F$ is smaller than $|Q|$ and when the number of amplification steps required is low (e.g. short decoding frames).

The convolutional codes defined below were chosen as a concrete and familiar example, rather than the most appropriate to the algorithm. On the other hand, their code lattice has a low `fanout' $F$ and hence showcases the QVA's superior performance in $|Q|$.
If using the maximal amount of parallelism, the QVA decodes the convolutional code in time $O(\sqrt{F^N} N \log F)$.
One can see the advantage is most apparent for $F \ll Q$ and for short decode frames $N$. Some general HMMs with absorbing states
will also show similar advantages.
 
In some applications of the QVA multiple trials need to be performed. That is, the problem is set up and the algorithm is iterated
the proper number of times, the output is measured and stored, and the process is repeated to accumulate classical statistical data.
These trials can be performed in parallel. For decoding the convolutional codes in this paper, the number of trials needed depends on the number of errors to be corrected.

We also analyze a variant of the QVA (the probabilistic QVA) in section \ref{sect:trials} with multiple trials and no iterations - the probability amplitudes are loaded into the paths during the lattice building step.  The probability of selecting the wrong answer (not finding the mode) falls exponentially in the number of trials.
There are applications where using the probabilistic QVA is advantageous -- but for decoding it seems that the QVA gives square root better performance than the probabilistic QVA.
	
	The first key idea of the QVA is to represent a particular path through the lattice, say $\pi = \pi_1 \pi_2 \cdots \pi_N$
	as the quantum state $\ket{\pi_1 \pi_2 \cdots \pi_N}$, where each $\ket{\pi_i} \in \C^Q$. This algorithm is a parallel
	algorithm; all of the approximately $\C^F$ admissible paths exist in superposition with equal amplitudes and with
	phases according to their probabilities as the algorithm builds the
	lattice. So the following lattice (here $Q=\{0,1,2\}$)
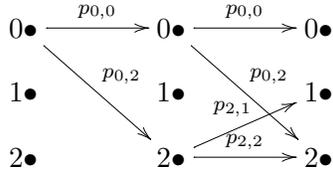
	\[
		\xymatrix@C=3.5em @R=1.0em{
			{0 \bullet} \ar[rdd]^{p_{0,2}} \ar[r]^{p_{0,0}} & {0 \bullet} \ar[rdd]^{p_{0,2}} \ar[r]^{p_{0,0}} & {0 \bullet} \\
			{1 \bullet}                                     & {1 \bullet}                                     & {1 \bullet} \\
			{2 \bullet}                                     & {2 \bullet} \ar[ru]^{p_{2,1}} \ar[r]^{p_{2,2}}  & {2 \bullet} }
	\]
	is represented by the unnormalized quantum state
	\begin{align}\label{eq:rqs}
		\ket{\psi_2} = 
			e^{i f(p_{0,0})f(p_{0,0})}\ket{000}  & + e^{i f(p_{0,0})f(p_{0,2})}\ket{002} \nonumber \\
				    & + e^{i f(p_{0,2})f(p_{2,1})}\ket{021} + e^{i f(p_{0,2})f(p_{2,2})}\ket{022},
	\end{align}
	where $f$ is some strictly increasing function (for example, the logarithm). We will call the operation that builds the lattice
$H_{\phi}$ and the operation that marks the states $G_{\phi}$. 

We will write $\mathcal{F}^N$ the space
of paths $\pi$ that are admissible given a sequence of emissions, so that $|\mathcal{F}^N| \approx F^N$.

The second key idea is to use a variant of amplitude amplification to shift the phases from the states into their
amplitudes. By analogy to Grover's algorithm and to $G_{\phi}$, we will call this step $G_{\mathcal{F}^N}$.
This procedure is reminiscent of both the Grover-based function minimization algorithms \cite{gromin} and the
multiphase kickback scheme of Meyer and Pommersheim \cite{MP}.
Multiplying \eqref{eq:rqs} by a global phase so that the most probable state (having the largest of the $f(p_{i,j}p_{j,k})$) has phase $-1$ gives lesser probable states a phase closer to $1$ if the monotonic function $f$ is chosen correctly. For
decoding convolutional codes this function corresponds to the number of errors that must have occurred for that path to have been taken.
In many cases, amplitude amplification still works even though the `unmarked' states do not have phase equal to $1$. This
can take up to $\pi/4 \sqrt{L}$ lattice building and amplifications steps, where $L = F^N$ is the number of states to search,
as in Grover's algorithm, but in the cases we will discuss below it may be fewer. 
\section{The Algorithm}
	As mentioned in the introduction, the algorithm first builds a quantum state, each step adding a tensor factor sequentially in
	the forward direction.
 
	\subsection{The building blocks of $H_{\phi}$ and $G_{\phi}$}
We start with an HMM of $|Q|$ states, where the state at time step $1 \le n \le N$ is written $x_n$, but is hidden; instead an emission $y_n \in Z$ is made visible by the HMM at the time the transition from $x_{n}$ to $x_{n+1}$ is made. 
The probability of the transition is written
\[ P(i,j|y) = \Pr(x_{n+1}=j|x_n=i, y_n=y), \]
with emission probabilities
\[ P(y|i,j) = \Pr(y_n=y | x_{n+1}=j, x_n=i). \]
We also let
\[ P_{i,j}(y) = P(i,j|y)P(y|i,j), \]
and clearly we have
\begin{equation}\label{eq:condclass}
   \sum_{j=1}^Q P(i,j|y) = 1 \text{ for all $y$}.
\end{equation}
If the additional condition
\begin{equation}\label{eq:conduni}
	\sum_j P_{i,j}(y) = 1 \text{ for all $i,y$}
\end{equation}
holds, then the transition probabilities could be stored in the state amplitudes instead of the phases, which
would allow an algorithm bypassing the amplification step. HMMs that have this condition \eqref{eq:conduni} however, are quite uncommon.

The unnormalized state $\ket{\psi_{k,y}}=\sum_j e^{i f(p_{k,j}(y))}\ket{j}$ represents a transition from state $k$ to state $j$ given $y$, with $F$ or fewer computational basis states in superposition. The $H_{\phi}G_\phi$ block links these transitions together into paths whose
total phase corresponds to the probability of the path being taken.

For a state $\ket{\psi}$, let $U_\psi$ be some unitary operation such that $U_\psi:\ket{0} \rightarrow \ket{\psi}$. There are many possible operators which satisfy that requirement, with a canonical (but generally nonoptimal) choice being given in section \ref{sect:canon}. 

The controlled operation which is active on the input pattern being $\ket{\psi_{c}}=\ket{k}$, implementing some unitary operation
$U:\C^Q \rightarrow \C^Q$ will be denoted
\begin{equation}\label{eq:link}
	\Qcircuit @C=1em @R=.7em {
\ket{\psi_{c}} & & {/} \qw & \measure{k} \qwx[1] & \qw \\
\ket{\psi_{t}} & & {/} \qw & \gate{U} & \qw \\
}
\end{equation}
That is, the above controlled operation takes $\ket{k}\ket{\psi_t}$ to $\ket{k}U\ket{\psi_t}$, extended by linearity, and acts as the identity
on all $\ket{j}\ket{\psi}$ for $j \ne k$.

For example, if $Q=2$ then the gate 
\[ \Qcircuit @C=1em @R=.7em {
  & \measure{1} \qwx[1] & \qw \\
  & \gate{U_{\ket{1}}} & \qw \\
} \]
is the standard two-qubit controlled not gate \cite{NC}, up to relative phase, since for a $2\times 2$ unitary
$U_{\ket{1}}:\ket{0} \rightarrow \ket{1}$ forces $U_{\ket{1}}:\ket{1} \rightarrow \alpha \ket{0}$.

We will let the gate
\begin{equation} \label{def:v}
	\Qcircuit @C=1em @R=.7em {
	& {/} \qw & \multigate{1}{V_y} & \qw \\
	& {/} \qw & \ghost{V_y} & \qw \\
	}
\end{equation}
stand for the string of controlled operations:
\[ \Qcircuit @C=1em @R=.7em {
& {/} \qw & \measure{0} \qwx[1] & \measure{1} \qwx[1] & \measure{2} \qwx[1] & \qw & {\cdots} &  & \qw & \measure{Q} \qwx[1] & \qw\\
& {/} \qw & \gate{U_{\ket{\psi_{0,y}}}} & \gate{U_{\ket{\psi_{1,y}}}} & \gate{U_{\ket{\psi_{2,y}}}} & \qw & {\cdots} & & \qw & \gate{U_{\ket{\psi_{Q,y}}}} & \qw \\
} \]
One can see that the matrix representation of the unitary operator $V_y:\C^Q \otimes \C^Q \rightarrow \C^Q \otimes \C^Q$ in the computational basis is block diagonal with the $U_{\ket{\psi_{i,y}}}$ for the blocks.
\subsection{The $G_{\phi}$ step}\label{sect:alg}
Given an initial distribution $\ket{\psi_0}$ and sequence of emissions $\{y_i\}_{i=1}^{N}$ perform
\begin{equation}\label{eq:gphi}
\Qcircuit @C=1em @R=.7em {
& \ket{\psi_0} &	&  \qw & \multigate{1}{V_{y_1}} & \qw &  \qw &\qw &\qw &\qw &\qw &\qw &\qw & \\
& \ket{0} & 	&  \qw & \ghost{V_{y_1}} & \multigate{1}{V_{y_2}} & \qw &\qw &\qw &\qw &\qw &\qw &\qw & \\
& \ket{0} &     & \qw  & \qw         & \ghost{V_{y_2}} & \multigate{1}{V_{y_3}} & \qw &\qw &\qw &\qw &\qw     &\qw   & \\
& \ket{0} &     & \qw  & \qw         &  \qw            & \ghost{V_{y_3}} & \qw & \cdots & &     &      &             &  \\
& \vdots  &     &      &             &                 &                 &     &        &  & & &                     &  \\ 
&         &  &     &     &     &     & & \cdots   &  & \multigate{1}{V_{y_{N-1}}}     & \qw                   &\qw   & \\
& \ket{0} &  & \qw &\qw & \qw & \qw & \qw & \qw    & \qw & \ghost{V_{y_{N-1}}}        & \multigate{1}{V_{y_N}} & \qw &  \\
& \ket{0} &  & \qw &\qw & \qw & \qw & \qw & \qw    & \qw & \qw                        &  \ghost{V_{y_N}}    & \qw    &  \\ }
\end{equation}
where the slashes denoting bundles of quantum wires are omitted for clarity. This construction of $G_\phi$ implements $N$ of the $V_y$ blocks,
each of which are composed of $Q$ controlled operations, for a gate complexity of $O(N Q F(\log F)^2)$.
The gate complexity can be improved by considering $G_\phi$ as sequence of $F^N$ commuting Grover marking operators $G_\pi$,
for each $\pi \in \mathcal{F}^N$, restricted to the subspace corresponding to $\mathcal{F}^N$ and marking with a
phase of $e^{i\prod_j f(\pi_j)}$ instead of $-1$. Each Grover operator can be implemented with gate complexity $\log(F^N)$ \cite{diao},
the product of which is preceded by an operation $H_{\mathcal{F}^N}$ which builds an equal superposition of states in $\mathcal{F}^N$.
The Grover inversion restricted to $\mathcal{F}^N$, 
\[G_{\mathcal{F}^N}=I-\frac{1}{\sqrt{F^N}}\sum_{\pi \in \mathcal{F}^N}\ket{\pi}\bra{\pi}, \]
and the operation $H_{\mathcal{F}^N}$ can also be implemented with $\log(F^N)$ gates. A concrete, but less efficient construction
of $G_{\mathcal{F}^N}$ follows the procedure outlined above for $G_{\phi}$.

Each of the $V_y$ are implicitly classically controlled, conditional on the emissions which is the code as received. 
To show that the construction \eqref{eq:gphi} implements $G_{\phi}$ we have:
\begin{proposition}
Fixing a set of emissions $\{y_i\}$,
The state of the system after one execution of \eqref{eq:gphi} is
\begin{equation}\label{eq:prop}
 \ket{\psi_f} = \sum_{\pi \in \mathcal{F}^N} e^{i \prod_j f(\pi_j)} \ket{\pi}, 
\end{equation}
where $\prod_j \pi_j$ is the probability of the path $\pi \in Q^N$ being taken.
\end{proposition}
\begin{proof}
Assume the assertion is true for $N-1$ steps, so
\[ \ket{\psi_{N-1}} = \sum_{\pi} e^{i \prod_j f(\pi_j)} \ket{\pi_1 \pi_2 \cdots \pi_{N-1}}, \]
with the sum over all possible paths $\pi$ of length $N-1$,
and then apply the gate
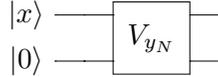
\[ \Qcircuit @C=1em @R=.7em {
 & \ket{x}&  &	 \qw & \multigate{1}{V_{y_N}} & \qw \\
& \ket{0}          &  &  \qw & \ghost{V_{y_N}} & \qw \\
	}
\]
with the last qubit of $\ket{\psi_{N-1}}$ being fed into the wire marked with $\ket{x}$. The output of that
operation will be 
\begin{align*}
	\ket{\psi_N} & = \sum_{\substack{i \text{ s.t. } \\ \pi_{N-1}=i}} e^{i \prod_{j=1}^{N-1} f(\pi_j)} \ket{\pi_1 \pi_2 \cdots \pi_{N-2}} \ket{\psi_{i,y_N}} \\
			     & = \sum_{\substack{i \text{ s.t. }\\ \pi_{N-1}=i}}\quad
							\sum_{\substack{j \text{ s.t. }\\ \pi_N = j}} e^{i \prod_{j=1}^{N-1} f(\pi_j)}\ket{\pi_1 \pi_2 \cdots \pi_{N-2} i}\ket{j}
								e^{i f(P_{i,j}(y_N))}                                                     \\
				 & = \ket{\psi_f}.
\end{align*} 
The base case follows from the definition: If $\pi_1$ is the path from $i$ to $j$ (given $y_1$), then $\pi_1 = P_{i,j}(y_1)$.
\end{proof}

\subsection{Implementation of a canonical unitary $U_\psi$}\label{sect:canon}
The classically controlled operation $H_{\mathcal{F}^N}$ is built up from quantum controlled $U_\psi$ gates where
$\ket{\psi}$ is a superposition of quantum states representing the fanout of a state in the HMM. 
For $q=\log |Q|$ and if $\ket{\psi}=\sum_{i=1}^{2^q}\psi_i \ket{i}$, with $K$ of the $\psi_i$ nonzero, we call $K$ the fanout of the state, and hence the maximum $K$ for all the $\psi$ derived from the HMM transitions is $F$, the fanout of the HMM.
When every state has the same fanout equal to $F$, $\ket{\psi}$ is always an equal superposition. Otherwise, there are cases when $U_\psi$ must take $\ket{0}$ to some $\ket{\psi} \in \R^{K+1}$.

For binary convolutional
codes defined in section \ref{sect:conv}, with $k$ the number of bits in the message block, $F =2^k$, the number of possible states that the encoder can transition to given a new message block entering the encoder. For the codes in that section, $Q$ can be made
arbitrarily large (while $F$ remains fixed) by increasing the constraint length.
 
Now we construct the sub-block $U_\psi$ from basic gates, in particular
\[ R_y(\theta) = \begin{pmatrix}\cos \frac{\theta}{2} & -\sin \frac{\theta}{2} \\ \sin \frac{\theta}{2} & \cos \frac{\theta}{2} \end{pmatrix}, \]
and then use generalized Toffoli gates to construct the controlled operations in $V_y$. The following method is similar to that in \cite{NC}.

If $R_{y,a,b}$ is the two-level unitary acting nontrivially on the subspace spanned by $\ket{a},\ket{b}$ then
constructing a unitary matrix $R$ with specified first column $\ket{\phi}$ will give $R:\ket{0} \rightarrow \ket{\phi}$.
	
		For some $2\times 2$ unitary matrix $W$, we write $W_{\{\ket{a},\ket{b}\}}$ for the two-level unitary operator which
		acts as $W$ on the subspace spanned by $\{\ket{a},\ket{b}\}$ and as the identity elsewhere. That is,
		\begin{align*}
		  W_{\{\ket{a},\ket{b}\}} \ket{a} & = (W)_{1,1}\ket{a} + (W)_{1,2} \ket{b}, \\
		  W_{\{\ket{a},\ket{b}\}} \ket{b} & = (W)_{2,1}\ket{a} + (W)_{2,2} \ket{b} \text{ and } \\
		  W_{\{\ket{a},\ket{b}\}} \ket{\psi} & = \ket{\psi},
		\end{align*}
		for $\ket{\psi}$ not in the span of $\{\ket{a},\ket{b}\}$.
		Let
		\[ R_y(\theta) = 
			\begin{pmatrix}\cos \frac{\theta}{2} & -\sin \frac{\theta}{2} \\ 
							\sin \frac{\theta}{2} & \cos \frac{\theta}{2} 
			\end{pmatrix} \] 
		be the y-rotation operator, or the exponentiated $Y$ gate. If the context is clear we will also write
		\begin{equation}R_y(t)=\begin{pmatrix}t & -\sqrt{1-t^2} \\  \label{def:ry}
							\sqrt{1-t^2} & t 
			\end{pmatrix}.
		\end{equation}
		We claim that the first column of
		\[ V := R_y(\theta_1)_{\{\ket{0},\ket{1}\}} R_y(\theta_2)_{\{\ket{0},\ket{2}\}} \cdots R_y(\theta_{N})_{\{\ket{0},\ket{N}\}} \]
		in the computational basis is the vector
		\begin{equation}\label{eq:firstcol}
			\begin{pmatrix} \cos(\theta_1)\cos(\theta_2)\cdots \cos(\theta_{K}) \\
						   \cos(\theta_2)\cos(\theta_3)\cdots \cos(\theta_{K})\sin(\theta_1) \\
						   \cos(\theta_3)\cos(\theta_4)\cdots \cos(\theta_{K})\sin(\theta_2) \\
						   \vdots \\
						   \cos(\theta_{K-1})\cos(\theta_{K-2})\sin(\theta_{K-1}) \\
						   \cos(\theta_{K})\sin(\theta_{K}) 
			\end{pmatrix},
		\end{equation}	
		which is equivalent to spherical coordinates in $\R^{K+1}$. Thus we have constructed a gate
		\[ V: \ket{0} \rightarrow \ket{\psi} \]
		for any $\ket{\psi} \in \R^{K+1}$, which is sufficient for our purposes.
		To see that this is true, assume that a product of $k$ $R_y$ matrices has the block form 
		\[ \prod_{j=1}^k R_y(\theta_j)_{\{\ket{0},\ket{j}\}} = 
					\left(\begin{array}{c|c|c}
						   &           & 0  \\
						 \vec{a} & \mbox{\huge $A$ } & \vdots  \\
						   &           & 0 \\
						\hline   
                        0  & 0 \ 0 \cdots 0 \ 0 & 1						
			\end{array}\right) \oplus I_{K-1-(k+1)}.\]
		Since $R_y(\theta_{k+1})_{\{\ket{0},\ket{k+1}\}}$ has the matrix form
		\[ \begin{pmatrix} \cos(\theta_{k+1}) & 0 & \cdots & 0 & -\sin(\theta_{k+1}) \\
							0                 & 1 &  0     & \cdots & 0  \\
							0 & 0 & 1 & 0 & 0\\
							\sin(\theta_{k+1}) & 0 & 0 & \cdots & \cos(\theta_{k+1})
							\end{pmatrix} \oplus I_{K-1-(k+1)}, \]
		we find the product
		\[ \prod_{j=1}^k R_y(\theta_j)_{\{\ket{0},\ket{j}\}} = 
					\left(\begin{array}{c|c|c}
						   &           &   \\
						 \cos(\theta_{j+1})\vec{a} & \mbox{\huge $A'$ } & A''  \\
						   &           &   \\
						\hline   
                        \sin(\theta_{j+1})  & 0 \ 0 \cdots 0 \ 0 & A'''						
			\end{array}\right) \oplus I_{K-1-(k+1)}.\]
	Thus using the base case: $\vec{a}=\left(\cos \theta_1,\sin \theta_1 \right)$ and the recursion for $\vec{a}$ we get \eqref{eq:firstcol}.
	Therefore it requires $K$ ($F$) rotation operators $R_y$ to implement $V$ in addition to the logic to change the subspaces that the $R_y$ act on. Using the Gray coding technique for the universal construction of quantum gates \cite{NC}, the subspace changing logic requires $f^2 2^f$, (where $F=2^f$) operations per $V$, or $F(\log F)^2)$. One can obtain even better results in the combinatorial control
logic by exploiting the special structure of $V$, but we will not need them here.

In the first amplification step of a trial we can combine $H_{\mathcal{F}^N}$ with $G_\phi$ by interlacing phase rotation operators $e^{i Z t_j}$ with the two-level unitaries $R_{y,a,b}$ above. To implement the classically controlled $G_\phi$ separate from
$H_{\mathcal{F}^N}$ as needed for subsequent amplification steps in the trial, a similar procedure to the above is implemented, using the
same Toffoli combinatorial logic.

\section{Decoding a rate $1/2$ convolutional code}\label{sect:conv}
%
	An $(n,k)$-convolutional code \cite{blahut} is a trellis code, which divides a datastream into \emph{message} blocks of length $k$ 
	and encodes them into \emph{code} blocks of length $n$. 
	We will limit our discussion to binary convolutional codes, in which the datastream is a string of bits.
	An encoder can be defined by its generator matrix $\mathbf{G}(x)$, a $k$ by $n$ matrix of polynomials in $\Z_2[x]$. If the
	encoder has a memory of the previous $m$ blocks, the generator polynomials have degree $m$ at most.
	The following diagram represents an encoder for a $(2,1)$-convolutional code with $m=2$.
	
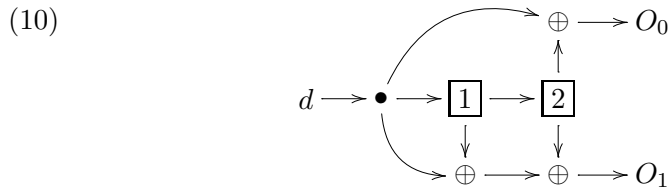
\begin{equation}\label{enc1}
		\xymatrix@C=1.5em @R=1.0em{
			 & & & {\oplus}\ar[r]& {O_0} \\
			d\ar[r]    & {\bullet} \ar@/^1.5pc/[urr] \ar[r] \ar@/_1pc/[dr]& {\boxed{1}} \ar[d] \ar[r] & \boxed{2} \ar[u] \ar[d]\\
			 &  & {\oplus} \ar[r] & {\oplus} \ar[r] & O_1 }
	\end{equation}
	The datastream is split into single bits and enters the encoder at the symbol $d$. At each discrete timestep the bit at $d$ is
	shifted right to the memory cell labeled $\boxed{1}$, and the previous contents of $\boxed{1}$ are shifted into $\boxed{2}$. The
	contents of $\boxed{2}$ are then discarded. The outputs $O_1$ and $O_2$ are formed by the various sums (in the field $\Z_2$) of
	memory contents and $d$. The generator matrix for encoder \eqref{enc1} is
	\[ \mathbf{G}(x) = [1+x^2 \quad 1+x+x^2], \]
	where the power of $x$ represents the time delay on the bit. Presumably the code words pass through a noisy channel and are
	received as \emph{receive} blocks by the decoder. Since the sequence of states (shift register contents) in the encoder can be
	modeled by an HMM, the classical VA can be used to implement the decoder module by tracing the
	most probable sequence of states for the encoder given the sequence of receive blocks.

	The state diagram
	for the code generated by \eqref{enc1} is
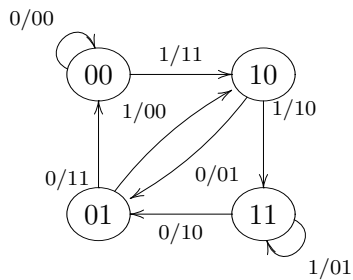
	\[
	\xymatrix@C=3.5em @R=3.0em{
		*++[o][F-]{00} \ar@(l,u)[]^{0/00} \ar[r]^{1/11} &  *++[o][F-]{10} \ar[d]^<<{1/10} \ar@/^/[ld]^{0/01}\\
		*++[o][F-]{01} \ar[u]^<<{0/11}  \ar@/^/[ru]^{1/00}           & *++[o][F-]{11} \ar@(r,d)[]^{1/01} \ar[l]^{0/10}}
	\]
	where the states are the boxed memory contents of the shift registers in the encoder. The transition to the
	next state is an arrow labeled by $i/o_0o_1$, where $i$ is the message bit and $o_0o_1$ are the code bits.
	
	To implement the algorithm, we consider a lattice segment representing the reception of the codeword $00$:
	
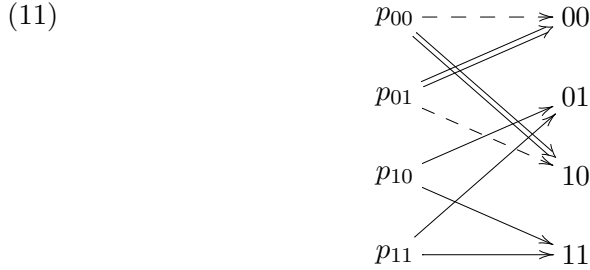
\begin{equation}\label{clrec}
		\xymatrix@C=4.5em @R=1.5em{
			{p_{00}} \ar@{=>}[rdd] \ar@{-->}[r] &00 \\
			{p_{01}} \ar@{=>}[ur] \ar@{-->}[dr] &01  \\
			{p_{10}} \ar[ur] \ar[dr]            &10  \\
			{p_{11}} \ar[uur] \ar[r]            &11 }
	\end{equation}
	with arrows representing
	\[\xymatrix@C=2em @R=0.0em{\text{$0$ errors} \ar@{-->}[r] &\text{, with phase rotation $p^0$,}\\ 
							  \text{$1$ error} \ar[r] & \text{, with phase rotation $p^1$,} \\ 
							  \text{$2$ errors} \ar@{=>}[r] & \text{, with phase rotation $p^2$, }\\} \]
	occurring in the channel. 
	The phase rotation $p=e^{i \omega}$ will be chosen so that there is a high probability of observing the most probable path
	after a fixed number of amplification steps. The idea is to make the paths with the fewest errors have a relative phase
	which is close to $-1$ while the less likely paths are made to have relative phase close to $1$.
	Under such circumstances, one would expect the algorithm to behave like a slightly noisy Grover search algorithm
	with the marked state being the most probable path. 

	For example, take $N=4$, so we have $Q=(\C^2)^4$ but we only need to search $\mathcal{F}^N=\C^4$. Supposing that no errors
	have occurred, the diagonal of $G_\phi$ looks like a permutation of
	\[ g_{\phi} = (p^0, p^2, p^3, p^3, p^3, p^4, p^4, p^4, p^4, p^4, p^5, p^5, p^5, p^5, p^6, p^7), \]
	where the permutation would be the identity in this case if the most probable path is $\ket{0000}$.
	Setting $\omega=0.68$ and amplifying three times ($\approx \pi/4 \sqrt{2^4}$) gives the state
	\[ (-0.76 + 0.29i)\ket{0000} + (0.16 - 0.05i)\ket{0001} + \cdots + (0.37 - 0.04 i)\ket{1111}, \]
	giving answer $\ket{0000}$ upon measurement with probability $\Pr_0=0.673$.
	
	The optimal phase $\omega^*$ varies with $N$ as the following table shows:
	
	\begin{tabular}{l|cccccccc}
	$N$        & 3    & 4    & 5    & 6    & 7    & 8    & 9  & 10 \\
	\hline
	iterations & 2    & 3    & 5    & 7    & 9    & 13   & 19 & 25\\
	\hline
	$\omega^*$ & 0.84 & 0.68 & 0.61 & 0.51 & 0.44 & 0.39 & 0.35&0.31\\
	\hline
	$\Pr_0$    & 0.73 & 0.67 & 0.73 & 0.76 & 0.76 & 0.79 & 0.82&0.73\\
	\end{tabular}
	
	The amplitude of the most probable state after approximately $\pi/4 \sqrt{F^N}$ iterations is a weighted sum
	of exponentials in independent variable $\omega$ with high frequency components increasing with $N$. One can gain
	intuition into choosing a good value for $\omega$ by 
    looking at the phase difference between the most probable state and the second most; it appears like damped oscillation.
	To get the most phase difference we want the maximum negative excursion. This must occur at the first local minimum.
	The oscillation frequency increases as the higher order terms are added so $\omega^*$ must decrease with $N$.
	
	Now suppose that one error has occurred. The diagonal of $G_\phi$ will look like a permutation of
	\[ g_{\phi_1} = (p^1, p^2, p^3, p^3, p^3, p^3, p^4, p^4, p^4, p^4, p^4, p^5, p^5, p^5, p^6, p^8) \] or
	\[ g_{\phi_2} = (p^1, p^3, p^3, p^3, p^3, p^4, p^4, p^4, p^4, p^4, p^4, p^5, p^5, p^5, p^6, p^6) \] or
	one of four other similar vectors, depending on where the error occurred.
	For these vectors, the maximal chance of success of observing $\ket{0000}$ tends to occur at an $\omega$ closer to $1$ than
	for the zero error case $g_{\phi}$. We will say that $g_{\phi}$ and $\{g_{\phi_1},g_{\phi_2}\}$ belong to
	different error classes.
Similarly, for two errors, the task of distinguishing vectors like
	\[ g_{\phi_3} = (p^2, p^2, p^2, p^2, p^3, p^3, p^3, p^4, p^4, p^5, p^5, p^5, p^5, p^6, p^6, p^7) \]
	is performed optimally for different $\omega$ and number of iterations. There are several ways to proceed
	in order to optimize the success of the algorithm for a given maximum number of errors corrected. One is
	the `detune' $\omega$; driving $\omega$ to a value closer to $0$ makes the probabilities of the possible paths
	closer to one another. The path probabilities are separated in different proportions as one changes the 
	number of iterations, which is possible another way to find a compromise between the error classes.
	
	We feel the best procedure is to adopt a sort of adaptive algorithm. For the first few trials (the number of
    trials needed is discussed below) $\omega$ is close
	to optimal for finding the paths corresponding to the most probable number of errors. If after those trials the
	expected mode is not achieved, $\omega$ is changed to correspond to the next most probable set of errors, and so on.
	Note that this procedure will introduce time complexity factor according to the number of probability classes. 
	Hence for these convolutional codes, this factor depends the maximum number of errors to be corrected.
	
	We can convert \eqref{clrec} to a unitary operation $V_{00}$ \eqref{def:v} using the representation as was
	obtained in \eqref{eq:rqs}. The operation $V_{00}$ then must map the following:
	
	\begin{align*}
	\ket{00}\ket{00} & \longmapsto \frac{1}{\sqrt{2}}\ket{00}(e^{i 0} \ket{00} + e^{2i\omega} \ket{01}) \\
	\ket{01}\ket{00} & \longmapsto \frac{1}{\sqrt{2}}\ket{01}(e^{2i\omega} \ket{00} + e^{i 0} \ket{01}) \\
	\ket{10}\ket{00} & \longmapsto \frac{1}{\sqrt{2}}\ket{10}(e^{i\omega} \ket{01} + e^{i\omega}\ket{11}) \\
	\ket{11}\ket{00} & \longmapsto \frac{1}{\sqrt{2}}\ket{11}(e^{i\omega} \ket{01} + e^{i\omega}\ket{11})
	\end{align*}
	
	which can be implemented by the following quantum circuit:
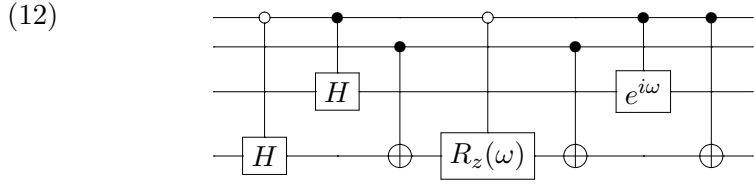
\begin{equation} \label{eq:qcirc} \Qcircuit @C=1em @R=.7em {  
	& \ctrlo{3} & \ctrl{2} &\qw      & \ctrlo{3}          & \qw      & \ctrl{2}           & \ctrl{3} & \qw\\
	& \qw       & \qw      &\ctrl{2} & \qw                & \ctrl{2} & \qw                & \qw & \qw \\
	& \qw       & \gate{H} & \qw     & \qw                & \qw      & \gate{e^{i\omega}} & \qw & \qw\\
	& \gate{H}  & \qw      &\targ    & \gate{R_z(\omega)} & \targ    & \qw                & \targ & \qw}  
\end{equation}
using the definition \eqref{def:ry}.

Notice that the two Hadamard gates implement a block of $H_{\mathcal{F}^N}$ (and hence are only called in the first iteration) while the remaining gates implement a block of $G_\phi$. The gate $R_z(\omega) = \begin{pmatrix}1 & 0\\0 & e^{i\omega}\end{pmatrix}$ and
the gate $e^{i\omega} = \begin{pmatrix}e^{i\omega}&0\\0&e^{i\omega}\end{pmatrix}$. Essentially the same operations are implemented in $V_{01}$, $V_{10}$ and $V_{11}$.
\subsection{Trial reduction via amplitude amplification}\label{sect:trials}
When condition \eqref{eq:conduni} holds, or by suitable choice of $f$, the results of section \ref{sect:canon}
can be used to put the probabilities of the paths into the amplitudes of the corresponding state. 
Thus with a single iteration of the algorithm the most probable path is most likely to be observed. 
One may ask if this probabilistic version of the QVA might be preferred over having to repeatedly apply the amplitude amplitication
iteration. The lemma below will tell us that for these convolutional codes, $O(F^N)$ trials of one iteration are needed to
have a good chance of finding the most probable path, in contrast to $O(\sqrt{F^N})$ iterations with $O(1)$ trials
for the QVA.

The trials for the probalistic QVA can be performed in parallel on an ensemble of $r$ quantum systems each containing $Q$ qubits. 
After performing the measurement step, we get $S$ samples of size $Q^N$.
Finding the most probable event is an instance of the multinomial selection problem \cite{ramey}. We use the following lemma from \cite{geroch}:
\begin{lemma}
Consider a collection of numbers with sum one. Denote the largest $b$ and the next largest by $b'$. Carry out $r$ runs in the
corresponding probability distribution and denote by $\kappa(r)$ the probability that the most frequent outcome is not the most
probable (i.e. the $b$-outcome). Then the following limit exists:
\[ \lim_{r\rightarrow \infty} \frac{-\log \kappa(r)}{r}= \lambda = \frac{(b-b')^2}{2[b(1-(b-b'))^2+b'(1+(b-b')^2)]}.\]
Hence $\kappa(r) \rightarrow 0$ as $e^{-\lambda r}$.
\end{lemma}
When no errors have occurred in the channel, $b=E_0^N$ and $b'=E_1 E_0^{N-1}$, for $E_0$ an appropriate value chosen for
no error occurring in a particular code word and $E_1$ a value chosen for one error. Then
\[\lambda = \frac{1}{2}\frac{E_0^N}{6-E_0^N}.\]
it is decided, for example, that a decode error probability of $e^{-2} \approx 0.13$ is acceptable, then we have
\[ \frac{0.8^N}{6-0.8^N} = \frac{4}{r}, \]
or
\[ r \sim 24 (1.25)^{N} - 4, \]
if $E_0=0.8$.
A similar procedure is used when collecting the results from the trials of the QVA.
Each $r$ trials are performed in parallel and then measured giving an array of length $r$ of strings of length $N \log Q$ . 
The array is sorted in $rN\log(Q)\log(rN\log Q)$ time and then the mode is
extracted by scanning the array\cite{skiena}. 

The probabilistic QVA seems to be most naturally applied to problems with a large $Q$ and small $N$. 
When there is a path with probability close to the most probable path, lots of extra trials are needed to
separate those out. 

The general analysis of the QVA is difficult, but we can see if a single iteration QVA behaves similarly to a single trial of the probabilistic QVA.

Let the diagonal of $G_\phi = (g_0,g_1,\cdots,g_L)$, with the $g_i$s roughly corresponding to the path probabilities and if
we consider the amplitude of the most probable path after a single application of $H_\phi G_\phi$ and $G_{\mathcal{F}^N}$,
we get 
\[ \mathrm{Pr}_0 = \frac{|g_0 (L-2) - 2 \sum_{i=1}^{L-1} g_i|^2}{L^3}, \]
since the first row of $G_{\mathcal{F}^N}$ is $(-\frac{l-2}{L},\frac{2}{L},\cdots,\frac{2}{L})$.
$\mathrm{Pr}_0$ is maximal, of course, when 
\begin{equation}\Arg{\sum_{i=1}^{L} g_i}=\pi+\Arg{g_0},\label{eq:arg}\end{equation}
which corresponds to $G_\phi = (-1,1,\cdots,1)$, the standard Grover diagonal.
For $L=2^N$ and with $g_i = 1$ for $i>0$,
the probability of measuring a state other than the most probable is then
\[ 1-\mathrm{Pr}_0 = \frac{(2^N-4)^2}{2^{2N}}, \]
which means that asymptotically a single iteration trial is only $O(1)$ times better than a random guess,
which is correct with probability $1/2^N$.
Thus it is best to either use multiple iterations or to put the probabilities directly into the amplitudes as when implementing
the probabilistic QVA. 

\section{Discussion}
A possibility for an improved algorithm is to preload the amplitudes of the states with a function depending on their probabilities using
a $H_\phi$ modified  so that the system is already rotated near the desired state with amplitude of $\ket{0}$
greater than the others and then apply some number of iterations of $G_\phi G_{\mathcal{F}^N}$.

Further research is needed to determine if the optimization step of the algorithm is optimal. Other interesting algorithms are possible
using phase kickbacks other than $-1$.

\bibliographystyle{plain}
\bibliography{qvit}	
\nocite{*}
\end{document}